\documentclass[a4paper,11pt,reqno]{amsart}

\usepackage{graphicx,amssymb,datetime,float,MnSymbol,currfile,tikz,hyperref,multirow}
\usepackage[square,comma,numbers]{natbib}
\usepackage[foot]{amsaddr}
\usetikzlibrary{arrows}
\usetikzlibrary{decorations.markings}

\newtheorem{theorem}{Theorem}

\def\ci{\!\perp\!}

\newcommand{\comments}[1]{}

\tikzset{tt/.style={decoration={
  markings,
  mark=at position .485 with {\arrow{>}},
  mark=at position .515 with {\arrow{<}}},postaction={decorate}}}

\begin{document}

\title[]{Simple yet Sharp Sensitivity Analysis for Any Contrast Under Unmeasured Confounding}

\author{Jose M. Pe\~{n}a$^1$}
\address{$^1$Link\"oping University, Sweden.}
\email{jose.m.pena@liu.se}


\begin{abstract}
We extend our previous work on sensitivity analysis for the risk ratio and difference contrasts under unmeasured confounding to any contrast. We prove that the bounds produced are still arbitrarily sharp, i.e. practically attainable. We illustrate the usability of the bounds with real data.
\end{abstract}

\maketitle

\section{Introduction}

When the estimation of the causal effect of a treatment on an outcome in an observational study is hindered by the presence of unmeasured confounding, one may resort to bounding it instead. This solution goes under the name of sensitivity analysis. The bounds are typically functions of some sensitivity parameters whose values are provided by the analyst. The sensitivity parameters usually represent the association of the unmeasured confounders with the exposure and outcome.

The sensitivity analysis method in the work \cite{DingandVanderWeele2016a}, hereafter DV, has received considerable attention in the literature. For instance, see the survey work \cite{Blumetal.2020} and the follow-up works \cite{DingandVanderWeele2016b,VanderWeeleandDing2017,VanderWeeleetal.2019,Sjolander2020,Sjolander2024}. Unfortunately, the work \cite{Sjolander2020} shows that the DV bounds are not sharp or attainable (i.e., logically possible). Moreover, the DV bounds only apply to the risk ratio and difference. The work \cite{Sjolander2024}, hereinafter AS, solves these two problems by deriving arbitrarily sharp (i.e., practically attainable) bounds for any contrast between the probabilities of the counterfactual outcome under exposure and non-exposure. Unlike the DV work that directly derives bounds for the contrast, the AS work first derives bounds for the counterfactual probabilities, which are then combined to obtain bounds for the contrast. Ensuring that the former are arbitrarily sharp ensures that the latter are also so. The AS bounds are based on the same sensitivity parameters as the DV bounds.

Our previous work \cite{Penna2020}, hereinafter JP, proposes an alternative to the DV and AS methods based on a different set of sensitivity parameters, which arguably the analyst may sometimes find easier to specify. Like the AS bounds and unlike the DV bounds, the JP bounds are arbitrarily sharp. However, like the DV bounds and unlike the AS bounds, the JP bounds only apply to the risk ratio and difference. In this note, we extend the JP bounds to any contrast while retaining their arbitrarily sharpness. To do so, we follow the AS approach described above.

The rest of the paper is organized as follows. First, we derive our bounds for any contrast and prove that they are arbitrarily sharp. Then, we illustrate their usability with some real and highlight some differences between them and the AS bounds. Finally, we close with some discussion.

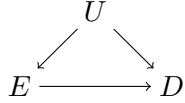
\begin{figure}[t]
\centering
\begin{tikzpicture}[inner sep=1mm]
\node at (0,0) (E) {$E$};
\node at (2,0) (D) {$D$};
\node at (1,1) (U) {$U$};
\path[->] (E) edge (D);
\path[->] (U) edge (E);
\path[->] (U) edge (D);
\end{tikzpicture}
\caption{Causal graph where $U$ is unmeasured.}\label{fig:graph}
\end{figure}

\section{Sharp Bounds}

Since our work is a follow-up of the JP work, we start by recalling some notation and derivations in that work. Consider the causal graph in Figure \ref{fig:graph}, where $E$ denotes the exposure, $D$ denotes the outcome, and $U$ denotes the set of unmeasured confounders. Let $E$ and $D$ be binary random variables. For simplicity, we assume that $U$ is a categorical random vector, but our results also hold for ordinal and continuous confounders. For simplicity, we treat $U$ as a categorical random variable whose levels are the Cartesian product of the levels of the components of the original $U$. We use upper-case letters to denote random variables, and the same letters in lower-case to denote their values. 

The causal graph in Figure \ref{fig:graph} represents a non-parametric structural equation model with independent errors, which defines a joint probability distribution $p(D,E,U)$. We make the usual positivity assumption that if $p(U=u) > 0$ then $p(E=e|U=u) > 0$, i.e., $E$ is not a deterministic function of $U$, and thus every individual in the subpopulations defined by the confounder can possibly be exposed or not \cite{HernanandRobins2020}.

Let $D_e$ denote the counterfactual outcome when the exposure is set to level $E=e$. Note that
\begin{align}\label{eq:D1}\nonumber
p(D_e=1) & = p(D_e=1 | E=e) p(E=e) + p(D_e=1 | E=1-e) p(E=1-e)\\
& = p(D=1 | E=e) p(E=e) + p(D_e=1 | E=1-e) p(E=1-e)
\end{align}
where the second equality follows from counterfactual consistency, i.e., $E=e \Rightarrow D_e = D$. We bound the counterfactual probability $p(D_e=1 | E=1-e)$ in terms of $p(D|E,U)$. Specifically,
\begin{align}\label{eq:D1E0}\nonumber
p(D_e=1 | E=1-e) & = \sum_u p(D_e=1 | E=1-e, U=u) p(U=u | E=1-e)\\
& = \sum_u p(D=1 | E=e, U=u) p(U=u | E=1-e)
\end{align}
where the second equality follows from $D_e \ci E | U$ for all $e$ in the causal graph in Figure \ref{fig:graph}, and counterfactual consistency. Thus,
\[
\min_{e,u} p(D=1 | E=e, U=u) \leq p(D_e=1 | E=1-e) \leq \max_{e,u} p(D=1 | E=e, U=u).
\]

Let us define two sensitivity parameters whose values the analyst has to specify:
\[
M=\max_{e,u} p(D=1 | E=e, U=u)
\]
and
\[
m=\min_{e,u} p(D=1 | E=e, U=u).
\]
By definition, these parameters' values must lie in the interval $[0,1]$ and $M \geq m$. The observed data distribution constrains the valid values further. To see it, note that
\[
p(D=1 | E=e) = \sum_u p(D=1 | E=e, U=u) p(U=u | E=e) \leq M
\]
for all $e$, and likewise
\[
p(D=1 | E=e) \geq m.
\]
Let us define
\[
M^*=\max_{e} p(D=1 | E=e)
\]
and
\[
m^*=\min_{e} p(D=1 | E=e).
\]
Then,
\[
M^* \leq M
\]
and
\[
m^* \geq m.
\]
We can thus define the feasible region for $M$ and $m$ as $M^* \leq M \leq 1$ and $0 \leq m \leq m^*$.

Incorporating the sensitivity parameters $M$ and $m$ in Equation \ref{eq:D1} leads to the following bounds of the counterfactual probability $p(D_e=1)$:
\begin{align}\nonumber
p(D=1, E=e) + p(E=1-e) m &\leq p(D_e=1)\\\label{eq:D1bound}
& \leq p(D=1, E=e) + p(E=1-e) M.
\end{align}
Note that setting the sensitivity parameters to their non-informative values $M=1$ and $m=0$ recovers the assumption-free bounds in the work \cite{Robins1989}.

At this point, our work departs from the JP work by adapting the AS approach for the DV sensitivity parameters to the JP sensitivity parameters. Concretely, we can obtain a lower (resp. upper) bound for any contrast between $p(D_1=1)$ and $p(D_0=1)$ by contrasting the lower (resp. upper) bound for $p(D_1=1)$ and the upper (resp. lower) bound for $p(D_0=1)$ in Equation \ref{eq:D1bound}. For instance, we can obtain bounds for the risk ratio, risk difference, odds ratio, odds difference, etc. It follows from the theorems below that these bounds are arbitrarily sharp for any contrast. This is therefore an extension of the JP result, which only applies to the risk ratio and difference.

\begin{theorem}\label{the:attainable1}
The lower bound for $p(D_1=1)$ and the upper bound for $p(D_0=1)$ in Equation \ref{eq:D1bound} are simultaneously arbitrarily sharp.
\end{theorem}

\begin{proof}
Let the set $\{M', m', p'(D,E)\}$ represent the sensitivity parameter values and the observed data distribution at hand. We assume that $M'$ and $m'$ belong to the feasible region. To show that the lower bound for $p(D_1=1)$ in Equation \ref{eq:D1bound} is arbitrarily sharp, we construct a distribution $p(D,E,U)$ that marginalizes to the set $\{M, m, p(D,E)\}$ such that (i) $\{M, m, p(D,E)\}$ and $\{M', m', p'(D,E)\}$ are arbitrarily close, and (ii) the lower bound and $p(D_1=1)$ are arbitrarily close.

Specifically,
\begin{itemize}
    \item let $p(E)=p'(E)$,
    \item let 
\begin{align*}
p(D=1|E=1,U=1)&=p'(D=1|E=1),\\
p(D=1|E=1,U=0)&=m',\\
p(D=1|E=0,U=1)&=M',\\
p(D=1|E=0,U=0)&=p'(D=1|E=0), \text{ and}
\end{align*}   
    \item let $U$ be binary with $p(U=1|E=1)=p(U=0|E=0)=1-\epsilon$ where $\epsilon$ is an arbitrary number such that $0<\epsilon<1$. The purpose of $\epsilon$ is to ensure that the positivity assumption holds.
\end{itemize}

Note that $M' \geq \max_e p'(D=1|E=e)$ and $m' \leq \min_e p'(D=1|E=e)$, because $M'$ and $m'$ belong to the feasible region. Then, $M=M'$ and $m=m'$.

Note also that
\begin{align*}
p(D=1|E=1) &= \sum_u p(D=1|E=1,U=u) p(U=u|E=1)\\
&=\epsilon \: m' + (1-\epsilon) \: p'(D=1|E=1)
\end{align*}
and thus $p(D=1|E=1)$ can be made arbitrarily close to $p'(D=1|E=1)$ by choosing $\epsilon$ sufficiently close to 0. Likewise for $p(D=1|E=0)$ and $p'(D=1|E=0)$.

Finally, recall from Equation \ref{eq:D1E0} that
\begin{align*}
p(D_1=1 | E=0) & = \sum_u p(D=1 | E=1, U=u) p(U=u | E=0)\\
& = (1-\epsilon) m' + \epsilon p'(D=1|E=1)
\end{align*}
which implies that $p(D_1=1 | E=0)$ can be made arbitrarily close to $m'$ and thus to $m$ by choosing $\epsilon$ sufficiently close to 0. Therefore, the lower bound for $p(D_1=1)$ in Equation \ref{eq:D1bound} can be made arbitrary close to $p(D_1=1)$ by Equation \ref{eq:D1}. That the upper bound for $p(D_0=1)$ in Equation \ref{eq:D1bound} is arbitrarily sharp can be proven analogously. 
\end{proof}

\begin{theorem}\label{the:attainable2}
The upper bound for $p(D_1=1)$ and the lower bound for $p(D_0=1)$ in Equation \ref{eq:D1bound} are simultaneously arbitrarily sharp.
\end{theorem}

\begin{proof}
The proof is analogous to that of Theorem \ref{the:attainable1} if we set $p(D=1|E=1,U=0)=M'$ and $p(D=1|E=0,U=1)=m'$ instead.
\end{proof}

\section{Example}

In this section, we illustrate the usability of our sensitivity analysis method with the real data in the AS work. The data originate from an observational study to estimate the causal effect of vitamin D insufficiency ($E=1$) on urine incontinence ($D=1$) in pregnant women \cite{Stafneetal.2020}. In particular, $p(E=1)=0.27$, $p(D=1|E=0)=0.38$ and $p(D=1|E=1)=0.49$.

\begin{table}[t]
\caption{Lower and upper bounds of the risk ratio as a function of the sensitivity parameters $M$ and $m$.}\label{tab:RR}
\centering
\footnotesize
\begin{tabular}{c|l|ccccc|}
\multicolumn{2}{c}{} & \multicolumn{5}{c}{$M$}\\\cline{3-7}
\multicolumn{2}{c|}{} & 0.49 & 0.62 & 0.75 & 0.87 & 1\\\cline{2-7}
\multirow{5}{*}{$m$} & 0.38 & (1.00,1.29) & (0.92,1.53) & (0.86,1.78) & (0.80,2.02) & (0.75,2.27)\\
& 0.29 & (0.83,1.38) & (0.77,1.65) & (0.71,1.91) & (0.66,2.17) & (0.62,2.43)\\
& 0.19 & (0.66,1.49) & (0.61,1.77) & (0.57,2.06) & (0.53,2.34) & (0.50,2.62)\\ 
& 0.1 & (0.49,1.62) & (0.45,1.92) & (0.42,2.23) & (0.39,2.54) & (0.37,2.85)\\
& 0 & (0.32,1.77) & (0.30,2.10) & (0.28,2.44) & (0.26,2.77) & (0.24,3.11)\\\cline{2-7}
\end{tabular}
\end{table}

\begin{table}[t]
\caption{Lower and upper bounds of the risk difference as a function of the sensitivity parameters $M$ and $m$.}\label{tab:RD}
\centering
\footnotesize
\begin{tabular}{c|l|ccccc|}
\multicolumn{2}{c}{} & \multicolumn{5}{c}{$M$}\\\cline{3-7}
\multicolumn{2}{c|}{} & 0.49 & 0.62 & 0.75 & 0.87 & 1\\\cline{2-7}
\multirow{5}{*}{$m$} & 0.38 & (0.00,0.11) & (-0.03,0.20) & (-0.07,0.30) & (-0.10,0.39) & (-0.14,0.48)\\
& 0.29 & (-0.07,0.14) & (-0.10,0.23) & (-0.14,0.32) & (-0.17,0.41) & (-0.21,0.51)\\
& 0.19 & (-0.14,0.16) & (-0.17,0.25) & (-0.21,0.35) & (-0.24,0.44) & (-0.28,0.53)\\
& 0.1 & (-0.21,0.19) & (-0.24,0.28) & (-0.28,0.37) & (-0.31,0.47) & (-0.35,0.56)\\
& 0 & (-0.28,0.21) & (-0.31,0.31) & (-0.35,0.40) & (-0.38,0.49) & (-0.42,0.58)\\\cline{2-7}
\end{tabular}
\end{table}

\begin{table}[t]
\caption{Lower and upper bounds of the odss ratio as a function of the sensitivity parameters $M$ and $m$.}\label{tab:OR}
\centering
\footnotesize
\begin{tabular}{c|l|ccccc|}
\multicolumn{2}{c}{} & \multicolumn{5}{c}{$M$}\\\cline{3-7}
\multicolumn{2}{c|}{} & 0.49 & 0.62 & 0.75 & 0.87 & 1\\\cline{2-7}
\multirow{5}{*}{$m$} & 0.38 & (1.00,1.57) & (0.87,2.28) & (0.76,3.41) & (0.66,5.44) & (0.57,10.22)\\
& 0.29 & (0.74,1.75) & (0.65,2.55) & (0.56,3.80) & (0.49,6.07) & (0.43,11.41)\\
& 0.19 & (0.54,1.96) & (0.47,2.86) & (0.41,4.26) & (0.35,6.81) & (0.31,12.79)\\
& 0.1 & (0.36,2.21) & (0.32,3.22) & (0.28,4.80) & (0.24,7.67) & (0.21,14.40)\\ 
& 0 & (0.22,2.50) & (0.19,3.64) & (0.17,5.44) & (0.14,8.68) & (0.13,16.31)\\\cline{2-7}
\end{tabular}
\end{table}

\begin{table}[t]
\caption{Lower and upper bounds of the odds difference as a function of the sensitivity parameters $M$ and $m$.}\label{tab:OD}
\centering
\footnotesize
\begin{tabular}{c|l|ccccc|}
\multicolumn{2}{c}{} & \multicolumn{5}{c}{$M$}\\\cline{3-7}
\multicolumn{2}{c|}{} & 0.49 & 0.62 & 0.75 & 0.87 & 1\\\cline{2-7}
\multirow{5}{*}{$m$} & 0.38 & (0.00,0.35) & (-0.10,0.79) & (-0.22,1.47) & (-0.36,2.72) & (-0.52,5.65)\\ 
& 0.29 & (-0.18,0.41) & (-0.28,0.85) & (-0.40,1.54) & (-0.54,2.78) & (-0.69,5.71)\\ 
& 0.19 & (-0.32,0.47) & (-0.43,0.91) & (-0.55,1.60) & (-0.68,2.84) & (-0.84,5.77)\\ 
& 0.1 & (-0.44,0.53) & (-0.55,0.96) & (-0.67,1.65) & (-0.80,2.90) & (-0.96,5.83)\\ 
& 0 & (-0.54,0.58) & (-0.65,1.01) & (-0.77,1.70) & (-0.90,2.95) & (-1.06,5.88)\\\cline{2-7}
\end{tabular}
\end{table}

Tables \ref{tab:RR}-\ref{tab:OD} show the lower and upper bounds for the risk ratio, risk difference, odds ratio and odds difference as a function of the sensitivity parameters $M$ and $m$. For instance, each entry of Table \ref{tab:OD} is of the form $(LB,UB)$ with
\[
LB=\frac{LB_1}{1-LB_1}-\frac{UB_0}{1-UB_0}
\]
and
\[
UB=\frac{UB_1}{1-UB_1}-\frac{LB_0}{1-LB_0}
\]
where $LB_e$ and $UB_e$ denote the lower and upper bounds of $p(D_e=1)$ in Equation \ref{eq:D1bound}. The R code for our analysis is available \href{https://www.dropbox.com/scl/fi/dgj5rpmmgf8gv9by4hh0y/sensitivityAnalysis4.R?rlkey=t1bcr6os0e59es0fp8uo016o1&dl=0}{here}. Recall that the bounds are arbitrarily sharp by Theorems \ref{the:attainable1} and \ref{the:attainable2}. Recall also that \cite{Penna2022} did not prove that the bounds for the odds ratio and odds difference are arbitrarily sharp. Recall also that the bounds for $M=1$ and $m=0$ coincide with the assumption-free bounds in \cite{Robins1989}.

The rest of this section highlights some differences between the AS and our bounds. Our sensitivity analysis method requires the analyst to describe the association between $U$ and $D$ with two parameters, whereas the AS method requires the analyst to describe the association between $E$ and $U$ with two parameters and the association between $U$ and $D$ with one parameters. Therefore, our method has one parameter less than the AS method. This simplifies the task of the analyst and helps visualization. However, a consequence of not describing the association between $E$ and $U$ is that our bounds always include the null causal effect, i.e., the undescribed association may be so strong as to nullify the causal effect. Note though that, for the risk and odds differences, our intervals are not necessarily centered at the null causal effect, and thus they are informative about both the magnitude and the sign of the true causal effect. Likewise for the risk and odds ratios. The AS bounds do not necessarily include the null causal effect. 

Finally, recall that our sensitivity parameters $M$ and $m$ are bounded as $M^* \leq M \leq 1$ and $0 \leq m \leq m^*$, because they are probabilities. This simplifies the task of the analyst and helps visualization. On the other hand, the AS sensitivity parameters are unbounded, because they are probability ratios. 

\begin{figure}[t]
\centering
\includegraphics[scale=.4]{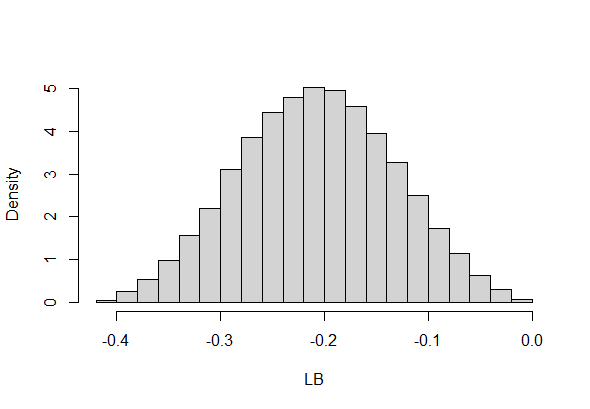}
\includegraphics[scale=.4]{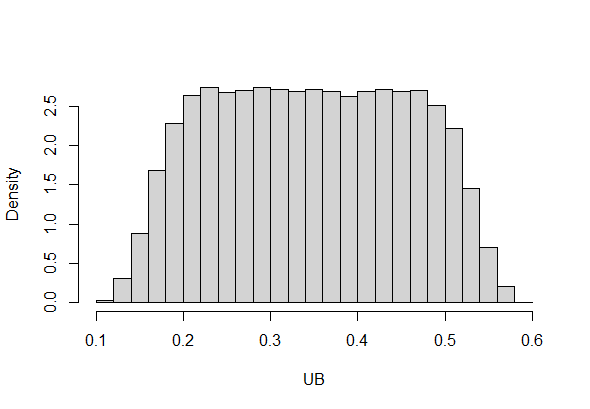}
\caption{Lower and upper bound distributions for the risk difference.}\label{fig:distributions}
\end{figure}

\section{Discussion}

In this note, we have extended our previous JP work on sensitivity analysis for the risk ratio and difference contrasts under unmeasured confounding to any contrast. We have proved that the bounds produced are still arbitrarily sharp. Our bounds are an alternative to the arbitrarily sharp bounds derived in the AS work for any contrast. The two alternatives are based on different sensitivity parameters. We believe that which set of parameters the analyst finds easier to specify may well depend on the domain under study. Therefore, we believe that no alternative is superior to the other.

An extension of our work that we are currently studying is the case where the analyst does not specify the values of the sensitivity parameters but their distribution. For instance, say that the analyst decides that $m$ follows a truncated normal distribution with mean $m^*/2$ and variance $0.1$ in the interval $(0,m^*)$. Likewise, say that the analyst decides that $M$ follows a uniform distribution in the interval $(M^*,1)$. We can now use Mote Carlo simulation to approximate the distributions of the lower and upper bounds for any contrast: Simply, sample a pair of values for $m$ and $M$, compute the bounds, and repeat. Figure \ref{fig:distributions} shows such distributions for the risk difference on the real data in the previous section. We can now use the approximated distributions to approximate the expectations of the bounds or the probabilities that the bounds are smaller or bigger than a given value.

\bibliographystyle{unsrt}
\bibliography{sensitivityAnalysis}

\end{document}